\documentclass[a4paper,UKenglish,cleveref, autoref, thm-restate]{lipics-v2021}

\title{The Complexity of \texorpdfstring{$\mathbf{(P_3, H)}$-Arrowing}{(P3, H)-Arrowing} and Beyond} 

\author{Zohair Raza {Hassan}}{Rochester Institute of Technology, Rochester, NY, 14623, USA}{zh5337@rit.edu}{
https://orcid.org/0000-0001-5590-5235}{}

\authorrunning{Z.\,R. Hassan} 

\Copyright{Zohair Raza Hassan} 



\category{} 

\relatedversion{} 

\acknowledgements{I want to thank my advisors Edith Hemaspaandra and Stanisław Radziszowski.}

\nolinenumbers 

\newcommand{\ra}{\rightarrow}

\newcommand{\np}{{\rm NP}}
\newcommand{\conp}{{\rm coNP}}

\newcommand{\epl}{\mathrm{epl}}
\newcommand{\mepl}{\mathrm{mepl}}

\newtheorem{problem}{Problem}

\begin{document}

\maketitle

\begin{abstract}
Often regarded as the study of how order emerges from randomness, Ramsey theory has played an important role in mathematics and computer science, giving rise to applications in numerous domains such as logic, parallel processing, and number theory. The core of graph Ramsey theory is arrowing: For fixed graphs $F$ and $H$, the $(F, H)$-Arrowing problem asks whether a given graph, $G$, has a red/blue coloring of the edges of $G$ such that there are no red copies of $F$ and no blue copies of $H$. For some cases, the problem has been shown to be coNP-complete, or solvable in polynomial time. However, a more systematic approach is needed to categorize the complexity of all cases.

We focus on $(P_3, H)$-Arrowing as $F = P_3$ is the simplest meaningful case for which the complexity question remains open, and the hardness for this case likely extends to general $(F, H)$-Arrowing for nontrivial $F$. In this pursuit, we also gain insight into the complexity of a class of matching removal problems, since $(P_3, H)$-Arrowing is equivalent to $H$-free Matching Removal. We show that $(P_3, H)$-Arrowing is coNP-complete for all $2$-connected $H$ except when $H = K_3$, in which case the problem is in P.
We introduce a new graph invariant to help us carefully combine graphs when constructing the gadgets for our reductions.
Moreover, we show how $(P_3,H)$-Arrowing hardness results can be extended to other $(F,H)$-Arrowing problems. This allows for more intuitive and palatable hardness proofs instead of ad-hoc constructions of SAT gadgets, bringing us closer to categorizing the complexity of all $(F, H)$-Arrowing problems.
\end{abstract}

\section{Introduction and related work}
\label{sec:intro}

At what point, if ever, does a system get large enough so that certain patterns become unavoidable? 
This question lies at the heart of Ramsey theory, which,
since its inception in the 1930s,
aims to find these thresholds for various combinatorial objects. Ramsey theory has played an important role in mathematics and computer science, finding applications in fields such as 
cryptography, algorithms, game theory, and more~\cite{rosta}.
A key operator within Ramsey theory is the arrowing operator, which is defined for graphs like so: 
given graphs $F, G$, and $H$, we say that $G \ra (F, H)$ (read, $G$ \textit{arrows} $F, H$) if every red/blue coloring of $G$'s edges contains a red $F$ or a blue $H$.
In this work, we analyze the complexity of computing this operator when $F$ and $H$ are fixed graphs. The general problem is defined as follows.

\begin{problem}
[$(F, H)$-Arrowing]
For fixed $F$ and $H$,
given a graph $G$, does $G \ra (F, H)$?
\end{problem}

The problem is clearly in coNP; a red/blue coloring of $G$'s edges with no red $F$ and no blue $H$ forms a certificate that can be verified in polynomial time since $F$ and $H$ are fixed graphs.
Such a coloring is referred to as an $(F, H)$-good coloring.
The computational complexity of $(F, H)$-Arrowing has been categorized for several---but not all---pairs $(F, H)$.
For instance, $(P_2, H)$-Arrowing is in P for all $H$, where $P_2$ is the path graph on $2$ vertices. This is because any coloring of an input graph $G$ that does not contain a red $P_2$ must be entirely blue, and thereby $(P_2, H)$-Arrowing is equivalent to determining whether $G$ is $H$-free.
Burr showed that $(F, H)$-Arrowing is coNP-complete when $F$ and $H$ are 3-connected graphs---these are graphs which remain connected after the removal of any two vertices, e.g., $(K_5, K_6)$-Arrowing~\cite{Bu3}. More results of this type are discussed in Section~\ref{sec:rw}.

In this work, we explore the simplest nontrival case for $F$, $F = P_3$, and provide a complete classification of the complexity when $H$ is a $2$-connected graph---a graph that remains connected after the removal of any one vertex. In particular, we prove:
\begin{theorem}
\label{thm:main}
    $(P_3, H)$-Arrowing is 
    coNP-complete for all $2$-connected $H$ except when $H = K_3$, in which case the problem is in P.
\end{theorem}

We do this by reducing an NP-complete SAT variant to $(P_3, H)$-Arrowing's complement,
$(P_3, H)$-Nonarrowing 
(for fixed $H$, does there exist a $(P_3, H)$-good coloring of a given graph?),
and showing how to construct gadgets for any $2$-connected $H$. It is important to note that combining different copies of $H$ can be troublesome; it is possible to combine graphs in a way so that we end up with more copies of it than before, e.g., combining two $P_n$'s by their endpoints makes several new $P_n$'s across the vertices of both paths.
Results such as Burr's which assume $3$-connectivity avoid such problems, in that we can
combine several copies of $3$-connected graphs without worrying about forming new ones. 
If $G$ is $3$-connected with vertices $u,v \in V(G)$ and we construct a graph $F$ by taking two copies of $G$ and identifying $u$ across both copies, then identifying $v$ across both copies, no new copies of $G$ are constructed in this process; if a new $G$ is created then it must be disconnected by the removal of the two identified vertices, contradicting $F$'s $3$-connectivity.
This makes it easier to construct gadgets for reductions.
To work with $2$-connected graphs and show how to combine them carefully, we present a new measure of intra-graph connectivity called \textbf{edge pair linkage}, and use it to prove sufficient conditions under which two copies of a 
$2$-connected graph $G$ can be combined without forming new copies of $G$. 

By targeting the $(P_3, H)$ case we gain new insight and tools for the hardness of $(F, H)$-Arrowing in the general case since $F = P_3$ is the simplest case for $F$. 
We conjecture that if $(P_3, H)$-Arrowing is hard, then $(F,H)$-Arrowing is also hard for all nontrivial $F$, but this does not at all follow immediately.
Towards the goal of categorizing the complexity of all $(F,H)$-Arrowing problems, we show how to extend the hardness results of $(P_3, H)$-Arrowing to
other $(F, H)$-Arrowing problems in Section~\ref{sec:extend}. 
These extensions are more intuitive and the resulting reductions are more palatable
compared to constructing SAT gadgets. 
We believe that techniques similar to the ones shown in this paper can be used to eventually categorize the complexity of 
$(F,H)$-Arrowing for all $(F,H)$ pairs.

The rest of the paper is organized as follows. Related work is discussed in Section~\ref{sec:rw}. 
We present preliminaries in Section~\ref{sec:prelim}, wherein we also define and analyze edge pair linkage.
Our complexity results for $(P_3, H)$-Arrowing are proven in Section~\ref{sec:p3complex}. We show how our hardness results extend to other arrowing problems in Section~\ref{sec:extend}, and we conclude in Section~\ref{sec:conclude}.
All proofs omitted in the main text are provided in the appendix.

\section{Related work}
\label{sec:rw}

\noindent \textbf{Complexity of $(F,H)$-Arrowing.} 
Burr showed that $(F,H)$-Arrowing is in P when $F$ and $H$ are both star graphs, or when $F$ is a matching~\cite{Bu3}. Hassan et al. showed that
$(P_3, P_3)$- and $(P_3, P_4)$-Arrowing are also in P~\cite{hassan2023}.
For hardness results, 
Burr showed that $(F, H)$-Arrowing is coNP-complete when $F$ and $H$ are members of $\Gamma_3$, the family of all $3$-connected graphs and $K_3$. 
The generalized $(F,H)$-Arrowing problem, where $F$ and $H$ are also part of the input, was shown to be $\Pi^p_2$-complete by Schaefer, who focused on constructions where $F$ is a tree and $H$ is a complete graph~\cite{Scha}.\footnote{$\Pi_2^p = \conp^{\np}$,
the class of all problems whose complements are solvable
by a nondeterministic polynomial-time Turing machine having
access to an NP oracle.} Hassan et al. recently showed that $(P_k, P_\ell)$-Arrowing is coNP-complete for all $k$ and $\ell$ aside from
the exceptions listed above~\cite{hassan2023}. We note that $(P_4, P_4)$-Arrowing was shown to be coNP-complete by Rutenburg much earlier~\cite{rut:c:graph-coloring}.

\noindent \textbf{Matching removal.} 
A matching is a collection of disjoint edges in a graph. 
Interestingly, there is an overlap between matching removal problems, defined below, and $(P_3, H)$-Arrowing. 

\begin{problem}[$\Pi$-Matching Removal~\cite{DBLP:journals/anor/LimaRSS22}]
Let $\Pi$ be a fixed graph property.
For a given graph $G$, does there exist a matching $M$ such that $G' = (V(G), E(G) - M )$ has property $\Pi$?
\end{problem}

Let $\Pi$ be the property that $G$ is $H$-free for some fixed graph $H$. Then, 
this problem is equivalent to $(P_3, H)$-Nonarrowing; a lack of red $P_3$'s implies that only disjoint edges can be red, as in a matching, and the remaining (blue) subgraph must be $H$-free.
Lima et al.\ showed that the problem is NP-complete when $\Pi$ is the property that $G$ is acyclic~\cite{lima2017decycling}, or that $G$ contains no odd cycles~\cite{lima2018bipartizing,DBLP:journals/anor/LimaRSS22}.

\noindent \textbf{Ramsey and Folkman numbers.}
The major research avenue involving arrowing is that of finding Ramsey and Folkman numbers. Ramsey numbers are concerned with the smallest complete graphs with arrowing properties, whereas Folkman numbers allow for any graph with some extra structural constraints.
We refer the reader to surveys by Radziszowski~\cite{ds1}
and Bikov~\cite{bikov2018} for more information on Ramsey and Folkman numbers, respectively.

\section{Preliminaries}
\label{sec:prelim}

\subsection{Notation and terminology}

All graphs discussed in this work are simple and undirected. $V(G)$ and $E(G)$ denote the vertex and edge set of a graph $G$, respectively. We denote an edge in $E(G)$ between $u,v \in V(G)$ as $(u,v)$.
For two disjoint subsets $A, B \subsetneq V(G)$, $E_G(A,B)$ refers to the edges with one vertex in $A$ and one vertex in $B$.
For a subset $A \subseteq V(G)$, $G[A]$ denotes the induced subgraph on $A$.
The neighborhood of a vertex $v \in V(G)$ is denoted as $N_G(v) = \{u ~|~ (u,v) \in E(G)\}$ and its degree as $d_G(v) := |N_G(v)|$.
A connected graph is called $k$-connected if it has more than $k$ vertices and remains connected whenever fewer than $k$ vertices are removed.

Vertex identification is the process of replacing two vertices $u$ and $v$ with a new vertex $w$ such that $w$ is adjacent to all remaining neighbors $N(u) \cup N(v)$.
For edges $(u,v)$ and $(p,q)$,
edge identification is the process of identifying $u$ with $p$, and $v$ with $q$.

The path, cycle, and complete graphs on $n$ vertices are denoted as $P_n$, $C_n$, and $K_n$, respectively. The complete graph on $n$ vertices missing an edge is denoted as $J_n$.
$K_{1,n}$ is the star graph on $n+1$ vertices.
For $n \geq 3$, we define $TK_n$ (tailed $K_n$) as the graph obtained by identifying a vertex of a $K_2$ and any vertex in $K_n$. The vertex of degree one in a $TK_n$ is called the tail vertex of $TK_n$.

We introduce a new notion, defined below, to measure how ``connected'' a pair of edges in a graph is, which will be useful when identifying edges between multiple copies of the same graph.
Examples have been shown in Figure~\ref{fig:mepl-eg}.

\begin{definition}
    For a pair of edges $e, f \in E(G)$, we define its \textbf{edge pair linkage}, 
$\epl_G(e,f)$, as the number of edges adjacent to both $e$ and $f$. It is infinity if $e$ and $f$ share at least one vertex. Note that $\epl_G(e, f) \leq 4$ when $e$ and $f$ share no vertices. 
For a graph $G$, we define $\mepl(G) := \min_{e, f \in E(G)} \epl(e,f) $ as the minimum edge pair linkage across all edge pairs.
\end{definition}

It is easy to see that the only graphs with $\mepl(G) = \infty$ are the star graphs, $K_{1,n}$, and $K_3$ since these are the only graphs that do not have disjoint edges.
When the context is clear, the subscript $G$ for $\epl_G(\cdot)$, $d_G(\cdot)$, etc.\ will be omitted. 

An $(F,H)$-good coloring of a graph $G$ is a red/blue coloring of $E(G)$ where the red subgraph is $F$-free, and the blue subgraph is $H$-free. We say that $G$ is $(F, H)$-good if it has at least one $(F,H)$-good coloring.
When the context is clear, we will omit $(F, H)$ and refer to the coloring as a good coloring. 

\subsection{Combining graphs}
\label{subsec:comb}

Suppose $H$ is a $3$-connected graph. Consider the graph $A_{H, e}$, obtained by taking two disjoint copies of $H$ and identifying some arbitrary $e \in E(H)$ from each copy. Observe that no new copy of $H$---referred to as a rogue copy of $H$ in $A_{H,e}$---is constructed during this process; if a new $H$ is created then it must be disconnected by the removal of the two identified vertices, contradicting $H$'s $3$-connectivity.
This is especially useful when proving the existence of good colorings; to show that a coloring of $A_{H, e}$ has no blue $H$, we know that only two copies of $H$ need to be looked at, without worrying about any other rogue $H$. 
Unfortunately, this property does not hold for all $2$-connected graphs. Instead, we 
use minimum edge pair linkage to explore sufficient conditions that allow us to combine multiple copies of graphs without concerning ourselves with any potential rogue copies. 

\begin{figure}
    \centering
    \includegraphics[width=0.95\textwidth]{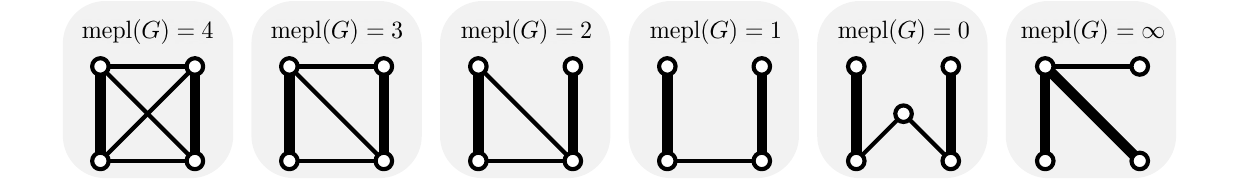}
    \caption{Graphs with different $\mepl(G)$ values. Bold edges have $\epl_G = \mepl(G)$.
    }
    \label{fig:mepl-eg}
\end{figure}
\begin{figure}
    \centering
    \includegraphics[width=0.95\textwidth]{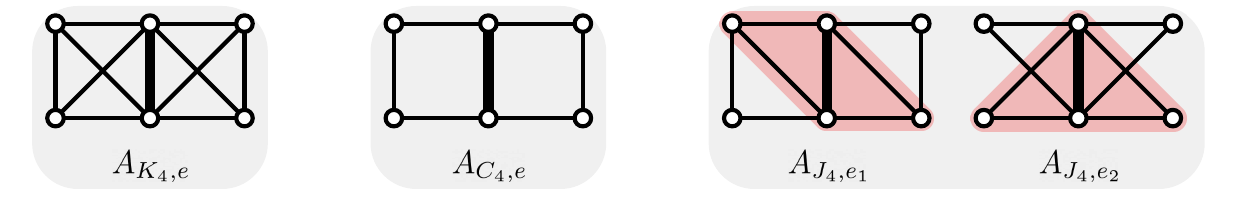}
    \caption{Proof for Lemma~\ref{lem:mlep2-combine} when $|V(H)| = 4$. It is easy to see that $A_{H,e}$ for $H \in \{C_4, K_4\}$ has exactly two copies of $H$ for arbitrary $e$. Moreover, constructing $A_{J_4, e}$ introduces a new $J_4$ (highlighted in red) for both nonisomorphic choices of $e \in E(J_4)$. Identified edges are bolded.
    }
    \label{fig:mepl-n4}
\end{figure}

\begin{lemma}
\label{lem:mlep2-combine}
Suppose $H$ is a $2$-connected graph such that $|V(H)| \geq 4$ and $\mepl(H) \geq 2$. 
Given $H$ and $e \in E(H)$, let $A_{H, e}$ be the graph obtained by taking two disjoint copies of $H$ and identifying $e$ from each copy. 
For all such $H$, except $J_4$, 
there exists $e \in E(H)$ such that $A_{H,e}$ has exactly two copies of $H$, i.e., no new copy of $H$ is formed after identifying $e$.
\end{lemma}

\begin{proof}
For $|V(H)| = 4$, the statement is easily observed as the only cases to consider are $C_4$, $J_4$, and $K_4$. See Figure~\ref{fig:mepl-n4}. 
Suppose $|V(H)| \geq 5$. 
We will first construct $A_{H,e}$ using an arbitrary $e = (u,v)  \in E(H)$ and assume that a new copy of $H$ is constructed after identifying $e$. 
Let $X$ and $Y$ denote the subgraphs of $A_{H,e}$ corresponding to the two copies of $H$ in $A_{H,e}$ that identify $(u,v)$. It follows that, $V(X) \cap V(Y) = \{u,v\}$ and $V(X) \cup V(Y) = V(A_{H,e})$. Similarly, $E(X) \cap E(Y) = \left\{ e \right\} $ and $E(X) \cup E(Y) = E(A_{H,e})$.
Suppose $Z$ is a subgraph corresponding to another copy of $H$ in $A_{H,e}$, i.e. $V(Z) \not= V(X)$ and $V(Z) \not= V(Y)$.
Let $V_{Z_X} = (V(X) - V(Y)) \cap V(Z)$ be the vertices of $Z$ only in $X$, and $E_{Z_X} = \left\{(p,q) \in E(Z) ~|~ p, q \in V_{Z_X} \right\}$.
$V_{Z_Y}$ and $E_{Z_Y}$ are defined similarly.
In the following claim, we observe the properties of $Z$ and the original graph $H$ with $\mepl(H) \geq 2$.
\begin{restatable}[]{claim}{meplclaim}
\label{claim:lem:mlep2-combine}
If $Z$ exists in $A_{H,e}$, the following must be true:
\textit{(1)} Both $V_{Z_X}$ and $V_{Z_Y}$ are nonempty, 
\textit{(2)} $u \in V(Z)$ and $v \in V(Z)$,
\textit{(3)} at least one of $E_{Z_X}$ and $E_{Z_Y}$ is empty, and 
\textit{(4)} there exists $w \in V(H)$ with $d_H(w) = 2$.
\end{restatable}
The proof of Claim~\ref{claim:lem:mlep2-combine} is given in Appendix~\ref{app:a}. 
Now, let $e = (u, v) \in E(H)$ such that $d_H(v) = 2$. Consider the graph $A_{H, e}$.
Note that since $d_H(v) = 2$, we have $d_{A_{H,e}}(v) = d_H(v) + d_H(v) - 1 = 3$. Let $w_X$ and $w_Y$ be the neighbors of $v$ in $V(X) - V(Y)$ and $V(Y) - V(X)$, respectively.
We know that $V(Z)$ includes $u$ and $v$ from Claim~\ref{claim:lem:mlep2-combine}(2). We now show that $w_X$ and $w_Y$ must also belong to $V(Z)$: if neither belong to $V(Z)$, then $d_Z(v) = 1$, contradicting $H$'s $2$-connectivity (removing $u$ disconnects $H$). Suppose w.l.o.g., that $w_X \not\in V(Z)$. Since $V_{Z_X}$ is nonempty and $H$ is connected, there is at least one vertex in $V_{Z_X}$ connected to $u$. However, removing $u$ would disconnect $Z$, again contradicting $H$'s $2$-connectivity. Thus, $\{u,v,w_X,w_Y\} \subseteq V(Z)$.
Using a similar argument, we can also show that both $(v,w_X)$ and $(v,w_Y)$ must belong to $E(Z)$.

Let $p$ be a vertex in $V(Z) - \{u,v,w_X,w_Y\}$. We know $p$ exists since $|V(H)| \geq 5$. W.l.o.g., we assume that $p \in V_{Z_X}$. Note that $p$ cannot be adjacent to $v$ since $N_{A_{H,e}}(v) = \{u,w_X, w_Y\}$.
We now consider the neighborhood of $p$ in $Z$. If $p$ has a neighbor $ q \in V_{Z_X} - \{w_X\}$, then $\epl_Z\left( (p, q), (v, w_Y)\right) = 0$, contradicting our assumption that $\mepl(H) \geq 2$. Since $d_Z(p) = d_H(p) \geq 2$ and the only options remaining for $p$'s neighborhood are $u$ and $w_X$, we must have that $p$ is connected to both $u$ and $w_X$. In this case, we have that $\epl_Z( (p, w_X), (v,w_Y) ) \leq 1$, which is still a contradiction. 
\end{proof}

\section{The complexity of \texorpdfstring{$\mathbf{(P_3, H)}$-Arrowing}{(P3, H)-Arrowing}}
\label{sec:p3complex}

In this section, we discuss our complexity results stated in Theorem~\ref{thm:main}. 
We first show that $(P_3, K_3)$-Arrowing is in P (Theorem~\ref{thm:p3k3-p}).
The rest of the section is spent setting up our hardness proofs for all $2$-connected $H \not= K_3$, which we prove formally in Theorems~\ref{thm:p3-hard-1} and~\ref{thm:p3-hard-2}. 
\begin{theorem}
\label{thm:p3k3-p}
$(P_3, K_3)$-Arrowing is in P.
\end{theorem}
\begin{proof}
    Let $G$ be the input graph. Let $\gamma : E(G) \ra \mathbb{Z}_{\geq 0}$ be a function that maps each edge to the number of triangles it belongs to in $G$. Note that $\gamma$ can be computed in $O(|E(G)|^3)$ time via brute-force. 
    Let $t$ be the number of triangles in $G$.
    Clearly, any $(P_3, K_3)$-good coloring of $G$ corresponds to a matching of total weight $\geq t$; otherwise, there would be some $K_3$ in $G$ in the blue subgraph of $G$.
    Now, suppose there exists a matching $M$ with weight at least $t$ that does not correspond to a $(P_3, K_3)$-good coloring of $G$.
    Then, there must exist 
    a copy of $K_3$ in $G$
    with at least two of its edges in $M$. 
    However, 
    since $K_3$'s maximal matching contains only 
    one edge, $M$ must contain a $P_3$, which is a contradiction. 
    Thus, we can solve $(P_3, K_3)$-Arrowing
    by finding the maximum weight matching of $(G, \gamma)$, which can be done in polynomial time~\cite{edmonds1965paths}, and checking if said matching has weight equal $t$.
\end{proof}

To show that $(P_3, H)$-Arrowing is coNP-complete we show that its complement, $(P_3, H)$-Nonarrowing, 
is NP-hard. We reduce from the following NP-complete SAT variant:

\begin{problem}[$(2,2)$-3SAT~\cite{berman200322sat}]
Let $\phi$ be a 3CNF formula where each clause has exactly three distinct variables, and each variable appears exactly four times: twice unnegated and twice negated. Does there exist a satisfying assignment for $\phi$?
\end{problem}

Important definitions and the idea behind our reduction are provided in Section~\ref{sec:reduc-spesh}, while the formal proofs are presented in Section~\ref{sec:reduc-proofs}.

\subsection{Defining special graphs and gadgets for proving hardness}
\label{sec:reduc-spesh}

We begin by defining a useful term for specific vertices in a coloring, after which we describe and prove the existence of some special graphs. We then define the two gadgets necessary for our reduction and describe how they provide a reduction from $(2,2)$-3SAT.

\begin{definition}
For a graph $G$ and a coloring $c$,
a vertex $v \in V(G)$ is called a \textbf{free vertex} if it is not adjacent to any red edge in $c$. Otherwise, it is called a \textbf{nonfree vertex.}
\end{definition}

Note that in any $(P_3, H)$-good coloring, a vertex can be adjacent to at most one red edge, otherwise, a red $P_3$ is formed. Intuition suggests that,
by exploiting this restrictive property, we could freely force ``any'' desired blue subgraph if we can ``append red edges'' to it. This brings us to our next definition, borrowed from Schaefer's work on $(F,H)$-Arrowing~\cite{Scha}:

\begin{definition}[\cite{Scha}]
A graph $G$ is called an \textbf{$(F, H)$-enforcer} with \textbf{signal vertex $v$} if 
it is $(F,H)$-good and 
the
graph obtained from $G$ by attaching a new edge to $v$ has the property that this edge
is colored blue in all $(F, H)$-good colorings.
\end{definition}

Throughout our text, when the context is clear,
we will use the shorthand \textbf{append an enforcer to $u \in V(G)$} to mean we will 
add an $(F, H)$-enforcer to $G$ and identify its signal vertex with $u$.
We prove the existence of $(F,H)$-enforcers when $F = P_3$ below. 
This proof provides a good example of the role $2$-connectivity plays while constructing our gadgets, showcasing how we combine graphs while avoiding constructing new copies of $H$. 
The arguments made are used frequently in our other proofs as well.

\begin{figure}[t]
    \centering
    \includegraphics[width=0.95\textwidth]{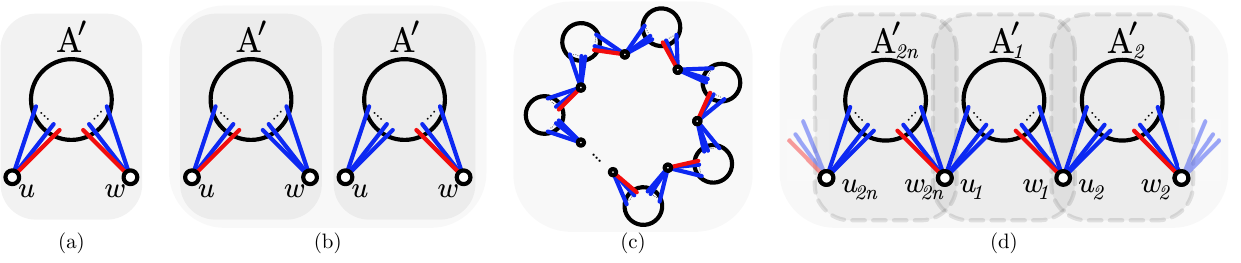}
    \caption{\textbf{(a)} The graph $A'$ when $u$ and $w$ are adjacent to a red edge.
    \textbf{(b)} The graph $A'$ when either $u$ or $w$ is adjacent to a red edge.
    \textbf{(c)} The graph $B$.
    \textbf{(d)} A zoomed in look at $B$.
    }
    \label{fig:enforcer}
\end{figure}

\begin{figure}[t]
    \centering
    \includegraphics[width=0.95\textwidth]{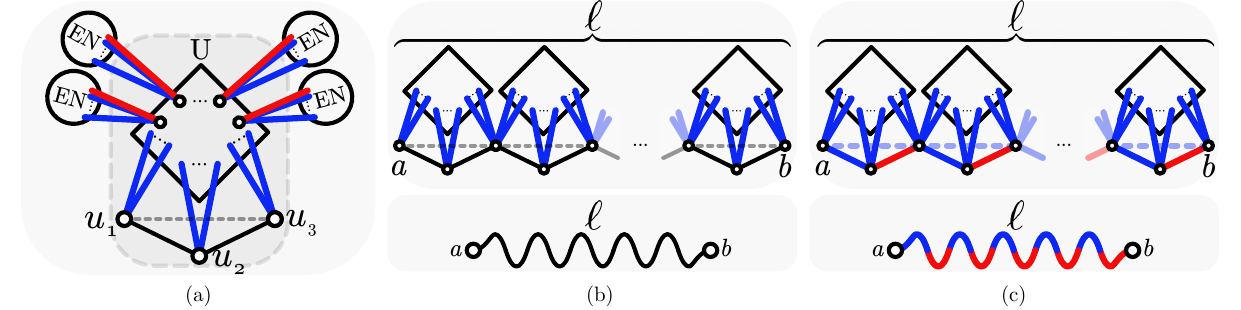}
    \caption{\textbf{(a)} The $(P_3, H)$-signal extender described in Lemma~\ref{lem:sig-exst}, where enforcers are labeled $EN$ and the copy of $H$ is labeled $U$. 
    Edges whose colors are fixed in all $(P_3, H)$-good colorings have been pre-colored. 
    The edge $(u_1,u_3)$ is dashed 
    to signify it may or may not exist in the construction. \textbf{(b)} At the top, we show how extenders can be connected sequentially to form arbitrarily large extenders. The enforcers have been removed from the illustration for clarity.
    The in- and out-vertices are marked $a$ and $b$, respectively. At the bottom, we show how signal extenders will be depicted in our figures,
    where $\ell$ is the number of 
    concatenated constructed 
    extenders. 
    \textbf{(c)} At the top, we show the coloring of the signal extender when vertex $a$ is a free vertex. At the bottom, we show the corresponding coloring of our representation of signal extenders.
    }
    \label{fig:signal-extender}
\end{figure}

\begin{lemma}
\label{lem:p3enforce}
$(P_3, H)$-enforcers exist for all $2$-connected $H$.
\end{lemma}
\begin{proof}
We extend an idea presented by Burr~\cite{Bu3}.
Let $A$ be a ``minimally bad'' graph such that $A \ra (P_3, H)$, but removing any edge $e$ from $A$ gives a $(P_3,H)$-good graph. 
Let $e = (u, w)$ and $A' = A - e$. 
This graph is illustrated in Figure~\ref{fig:enforcer}.
Observe that in any $(P_3,H)$-good coloring of $A'$, at least one edge adjacent to $u$ or $w$ must be red; otherwise, such a coloring and a red $(u,w)$ gives a good coloring for $A$, contradicting the fact that $A \ra (P_3, H)$. 
If both $u$ and $w$ are
adjacent to red edges in all good colorings, then $A'$ is a $(P_3, H)$-enforcer, and either $u$ or $w$ can be the signal vertex .

If there exists a coloring where only one of $\{u,w\}$ is adjacent to a red edge, then we can construct an enforcer, $B$, as follows.
Let $n = |V(H)|$.
Make $2n$ copies of $A'$, where $u_i$ and $w_i$ refer to the vertex $u$ and $w$ in the  $i^{\text{th}}$ copy of $A'$, called $A'_i$.
Now, identify each $w_i$ with $u_{i+1}$ for $i \in \{1,2,\ldots,2n-1\}$, and identify $w_{2n}$ with $u_1$ (see Figure~\ref{fig:enforcer}). 
Although $w_i$ and $u_{i+1}$ are now the same vertex in $B$, we will use their original names to make the proof easier to follow.
It is easy to see that when $w_1$ is adjacent to a red edge in $A'_1$, then $u_2$ cannot be adjacent to any red edge in $A'_2$, causing $w_2$ to be adjacent to a red edge in $A'_2$, and so on. A similar argument holds when considering the case where $u_1$ is adjacent to a red edge in $A'_1$.
Since every $u_i$ and $w_i$ is adjacent to a red edge, any of them can be our desired signal vertex.

Note that $B$ must be $(P_3, H)$-good because each $A'_i$ is $(P_3, H)$-good, and no new $H$ is made during the construction of the graph; since $H$ is $2$-connected, $H$ cannot be formed between two copies of $A'_i$'s, otherwise there is a single vertex that can be removed to disconnect such an $H$, contradicting $2$-connectivity. Thus, any new copy of $H$ must go through all $A_i$'s, which is not possible since such an $H$ would have $\geq 2|V(H)|$ vertices.
\end{proof}
Using enforcers, we construct another graph that plays an important role in our reductions.

\begin{definition}
A graph $G$ is called a \textbf{$(P_3, H)$-signal extender} with \textbf{in-vertex $a$} and \textbf{out-vertex $b$} if 
it is $(P_3,H)$-good and, in all $(P_3, H)$-good colorings,
$b$ is nonfree if $a$ is free.
\end{definition}

\begin{lemma}
\label{lem:sig-exst}
$(P_3, H)$-signal extenders exist for all $2$-connected $H$.
\end{lemma}
\begin{proof} 
Let $n = |V(H)|$.
    Construct a graph $G$ like so. Take a copy of $H$, and let $u_1, u_2, \ldots, u_{n} \in U$ be the vertices of $H$ in $G$, such that
    $(u_1,u_2)$ and $(u_2,u_3)$ are edges of $H$.
    For each $u_i$ for $i \in \{4, 5, \ldots, n \}$, append an enforcer to $u_i$.
    Observe that no ``new'' $H$ is constructed during this process since $H$ is $2$-connected.
    Since each vertex except $u_1, u_2, $ and $u_3$ is connected to an enforcer, each edge in $G[U]$, except $(u_1, u_2)$, $(u_2, u_3)$, and $(u_1, u_3)$ must be blue. However, not all of them can be blue, otherwise $G[U]$ is a blue $H$. 
    Therefore, in any good coloring, 
    if $u_1$ is a free vertex, $(u_2, u_3)$ must be red, making $u_3$ a nonfree vertex. Thus, $u_1$ is our in-vertex, and $u_3$ is our out-vertex.
    We illustrate $G$ in Figure~\ref{fig:signal-extender}(a).
\end{proof}

Observe that multiple copies of these extenders can be used to form larger ones (see Figure~\ref{fig:signal-extender}).
With the enforcer and extender graphs defined, 
we are ready to construct the gadgets for our
reductions. Below, we define variable and clause gadgets and describe how they are used in our proofs.
Recall that we are reducing from $(2,2)$-3SAT. We will explain how our graphs encode clauses and variables after the definitions.

\begin{definition}
For a fixed $H$,
a \textbf{clause gadget} is a $(P_3,H)$-good graph $CG$ containing vertices $i_1,i_2,$ and $i_3$---referred to as \textbf{input vertices}, such that if vertices outside the gadget, $o_1$, $o_2$ and $o_3$, are connected to $i_1$, $i_2$, and $i_3$, respectively, then each of the eight possible combinations of $(o_j, i_j)$'s colors should allow a $(P_3, H)$-good coloring for $CG$, except the coloring where all $(o_j, i_j)$'s are red, which should not allow a good coloring of $CG$.
\end{definition}

\begin{definition}
For a fixed $H$,
a \textbf{variable gadget} is a $(P_3, H)$-good graph $VG$ containing four \textbf{output vertices}: two \textbf{unnegated output 
vertices}, $u_1$ and $u_2$, and two \textbf{negated output 
vertices},
$n_1$ and $n_2$, such that:
\begin{enumerate}
    \item In each $(P_3,H)$-good coloring of $VG$:
        \begin{enumerate}
            \item If $u_1$ or $u_2$ is a free vertex, 
            then $n_1$ and $n_2$ must be nonfree vertices.
            \item If $n_1$ or $n_2$ is a free vertex, 
            then $u_1$ and $u_2$ must be nonfree vertices.
        \end{enumerate}
    \item There exists at least one $(P_3,H)$-good coloring of $VG$ where $u_1$ and $u_2$ are free vertices.
    \item There exists at least one $(P_3,H)$-good coloring of $VG$ where $n_1$ and $n_2$ are free vertices.
\end{enumerate}
\end{definition}

Note how
clause gadgets and their input vertices correspond to OR gates and their 
inputs; 
the external edges, denoted as $(o_j,i_j)$'s in the definition, behave like true or false signals:
blue is true, and red is false.
Similarly, output vertices of variable gadgets behave like
sources for these signals.
The reduction is now straightforward: for a formula $\phi$, construct $G_\phi$ like so.
For each variable and clause, add a corresponding gadget. Then, identify the output and input vertices according to how they appear in each clause. 
It is easy to see that 
any satisfying assignment of $\phi$
corresponds to a $(P_3, H)$-good coloring of $G_\phi$, and vice versa.
To complete the proof
we must show that the gadgets described exist and that no new $H$ is formed while combining the gadgets during the reduction.

Note that it is possible for a variable gadget to have a coloring where both unnegated and negated output vertices are nonfree. This does not affect the validity of the gadgets. A necessary restriction is that if variable $x$ appears unnegated in clause $L_1$ and negated in $L_2$, then $x$ cannot satisfy both clauses.
Our gadgets clearly impose that restriction.

\begin{figure}[t]
    \centering
    \includegraphics[width=0.95\textwidth]{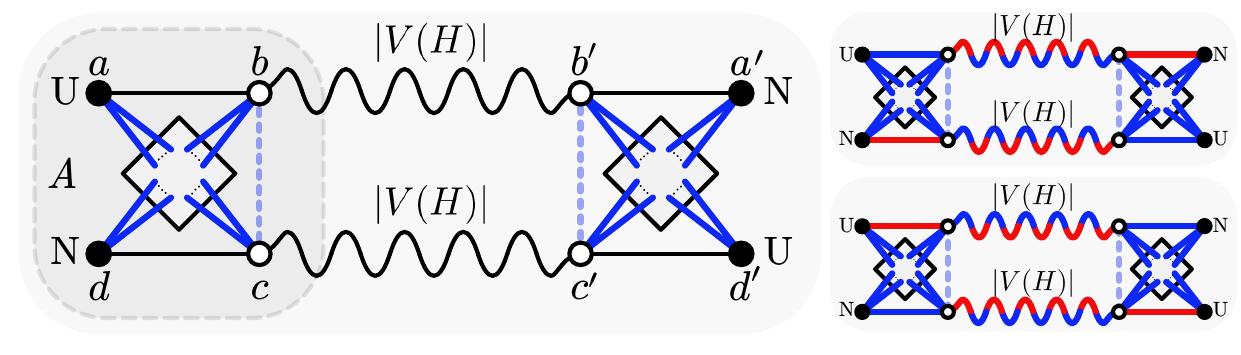}
    \caption{
    On the left, we show the variable gadget constructed using two copies of $A$ as described in Theorem~\ref{thm:p3-hard-1}'s proof and signal extenders. 
    The vertices in the square are the vertices of $A$ which had enforcers appended to them. The edge $(b,c)$ is dashed to signify that it may or may not exist.
    Edges have been precolored wherever possible. Note that if $(b,c)$ exists, it must be blue: if $(b,c)$ is red, the attached signal extenders will force 
    $(a',b')$, $(b',c')$, and $(c',d')$ to be blue, forming a blue $H$.
    By symmetry, $(b',c')$ must also be blue.
    Now, observe that at least one edge in $\{(a,b), (c,d)\}$ must be red, otherwise we form a blue $H$ in $A$. Suppose $(a,b)$ is red: the signal extender forces $(a',b')$ to be blue. To avoid a blue $H$, $(c',d')$ must be red, which forces $(c,d)$ to be blue. 
    In this case, the vertices marked $\mathbf{U}$ are nonfree vertices, and the vertices marked 
    $\mathbf{N}$ are free.
    A similar pattern can be observed when we color $(c,d)$ red instead, giving us colorings where vertices marked $\mathbf{U}$ are free and vertices marked $\mathbf{N}$ are nonfree. 
    }
    \label{fig:vg-1}
\end{figure}

\begin{figure}[t]
    \centering
    \includegraphics[width=0.95\textwidth]{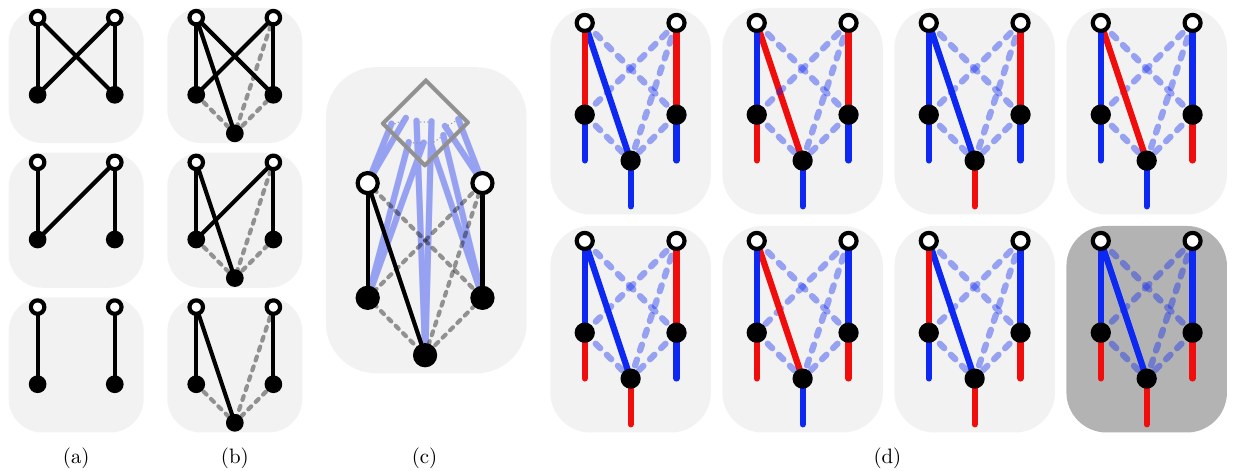}
    \caption{
    This figure shows the clause gadget used in the proofs of Theorem~\ref{thm:p3-hard-1} and~\ref{thm:p3-hard-2}.
    \textbf{(a)} Each block represents the induced subgraph in a $H$ when: \textit{(1)} $\mepl(H) = 2$ and $H$ has an induced $C_4$, \textit{(2)} $\mepl(H) = 1$, and \textit{(3)} $\mepl(H) = 0$.
    \textbf{(b)} Each block represents how a fifth vertex, denoted $e$ in the proofs of Theorems~\ref{thm:p3-hard-1} and~\ref{thm:p3-hard-2}, may be connected to the induced subgraphs from (a). 
    For each case, the solid line going from $e$ represents an edge that must exist in $E(H)$ since $d_H(v) \geq 2$ for all $v \in H$ due to $2$-connectivity. 
    In the first case, any edge can be chosen w.l.o.g.\ due to the symmetry of $C_4$.
    The dashed edges may or may not exist, but their existence is inconsequential to the correctness of our gadget. 
    \textbf{(c)} An illustration of the clause gadget, where each vertex of $H$ attached to an enforcer is in the square. The input vertices have been filled in.
    \textbf{(d)} 
     We show the eight possible combinations of inputs that can be given to the gadget. Observe that a $(P_3, H)$-good coloring is always possible unless the input is three red edges.
    }
    \label{fig:cg-1}
\end{figure}

\subsection{Hardness proofs}
\label{sec:reduc-proofs}

Using the ideas and gadgets presented in Section~\ref{sec:reduc-spesh}, we provide our hardness results below.

\begin{theorem}
\label{thm:p3-hard-1}
$(P_3, H)$-Arrowing is coNP-complete when $H$ is a $2$-connected graph on at least four vertices with $\mepl(H) \leq 1$.
\end{theorem}
\begin{proof}
    We reduce $(2, 2)$-3SAT to $(P_3, H)$-Nonarrowing as described in the end of Section~\ref{sec:reduc-spesh};
    given a $(2,2)$-3SAT formula $\phi$, we construct a graph $G_\phi$ such that $G_\phi$ is $(P_3,H)$-good if and only if $\phi$ is satisfiable.
    Since we have $\mepl(H) \leq 1$, we must have two edges $(a,b), (c,d) \in E(H)$ that have at most one edge adjacent to both. We construct a graph $A$ like so: take a copy of $H$ and append an enforcer to each vertex of $H$ except $a,b,c,$ and $d$. We construct the variable gadget using two copies of $A$ joined by signal extenders, as shown in Figure~\ref{fig:vg-1}. The vertices labeled $\mathbf{U}$ (resp., $\mathbf{N}$) correspond to unnegated (resp., negated) output vertices.
    Note that there are no rogue $H$'s made during this construction.
    Recall that because of $2$-connectivity, a copy of $H$ cannot go through a single vertex. Thus, if a copy of $H$ other than the ones in $A$ and the signal extenders exists, it must go through both copies of $A$ and the signal extenders. However, this is not possible because by including the two signal extenders, we would have a copy of $H$ with at least $2|V(H)|$ vertices.

    We now describe our clause gadget.
    As observed in Section~\ref{subsec:comb}, there is no $2$-connected graph with $\mepl(\cdot) \leq 1$ on four vertices so we assume $|V(H)| \geq 5$. 
    Let $(a,b)$ and $(c,d)$ be the edges that achieve this $\epl(\cdot)$. 
    Let $e$ be a fifth vertex connected to at least one of $\{a,b,c,d\}$.
    We construct a clause gadget by 
    taking a copy of $H$ and 
    appending an enforcer to each vertex except $a,b,c,d,$ and $e$.
    In Figure~\ref{fig:cg-1}---which also includes a special case for $\mepl(H) = 2$ used in Theorem~\ref{thm:p3-hard-2}---we show how $e$ and two vertices from $\{a,b,c,d\}$ can be used as input vertices so that a blue $H$ is formed if and only if all three input vertices are connected to red external edges. 
    Observe that we can make the clause gadget arbitrarily large by attaching the out-vertex of a signal extender to each input vertex of a clause gadget. 
    We attach a signal extender with at least $|V(H)|$ vertices to each input vertex. This ensures that no copies of $H$ other than the ones in each gadget and extender are present in $G_\phi$.
\end{proof}

\begin{figure}[t]
    \centering
    \includegraphics[width=0.95\textwidth]{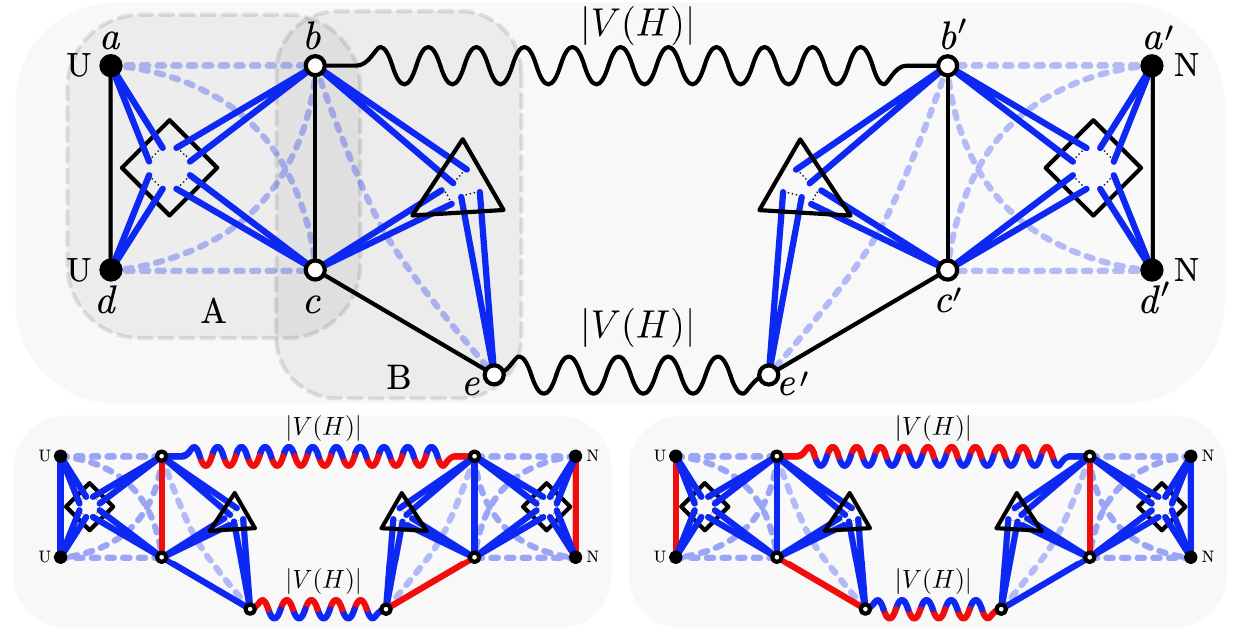}
    \caption{On the top, we show the variable gadget constructed using two copies of $F$ as described in Theorem~\ref{thm:p3-hard-1}'s proof and signal extenders. $A$ and $B$ have been marked.
    The vertices in the square (resp., triangle) are the vertices of $A$ (resp., $B$) which had enforcers appended to them. 
    Dashed edges signify edges that may or may not exist. 
    Edges have been precolored wherever possible.
    Observe that $(b,e)$ must be blue: if $(b,e)$ is red, the attached extenders will force $(b',c')$, $(b',e')$, and $(c',e')$ to be blue, forming a blue $H$ in the copy of $B$ on the right.
    We show that $(a,b)$, $(a,c)$, $(b,d)$, and $(c,d)$ must always be blue.
    Observe that at least one edge in 
    $\{(b,c), (c,e)\}$ must be red, otherwise we form a blue $H$ in $B$. 
    Note that if $(b,c)$ is red, the edges $(a,b)$, $(a,c)$, $(b,d)$, and $(c,d)$ must be blue. 
    If $(b,c)$ is blue, $(c,e)$ is red. Thus, $(c,d)$ and $(a,c)$ are blue.
    Moreover, a red $(c,e)$ forces 
    $(c',e')$ to be blue via the extender. Thus, 
    $(b',c')$ must be red to avoid a blue $H$. The extender on the top will in turn force edge $(a,b)$ and $(b,d)$ to be blue. 
    Therefore, $(a,b)$, $(a,c)$, $(b,d)$, and $(c,d)$ are blue in all good colorings.
    By symmetry, 
    $(a',b')$, $(a',c')$, $(b',d')$, $(c',d')$, and
    $(b',e')$ must also be blue in all good colorings.
    Observe that when a vertex marked $\mathbf{U}$ is nonfree, i.e., $(a,d)$ is blue, $(b,c)$ must be red. Thus, $(b',c')$ is blue, and $(a',d')$ must be red, making the vertices marked $\mathbf{N}$ nonfree. 
    A similar pattern can be observed when vertices marked $\mathbf{N}$ are free, wherein the vertices marked $\mathbf{U}$ are forced to be nonfree. These colorings are shown at the bottom of the figure.
    }
    \label{fig:vg-2}
\end{figure}
\begin{figure}[t]
    \centering
    \includegraphics[width=0.95\textwidth]{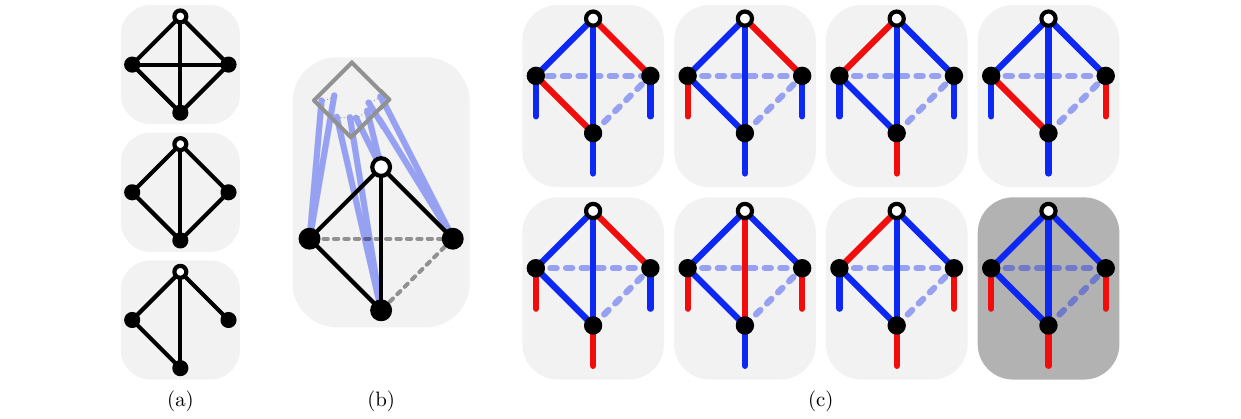}
    \caption{
    This figure shows the clause gadget used in Theorem~\ref{thm:p3-hard-2} when $H$ contains a $TK_3$.
    \textbf{(a)} Each block represents the induced subgraph in a $H$ when: \textit{(1)} $\mepl(H) = 4$, \textit{(2)} $\mepl(H) = 3$, and \textit{(3)} $\mepl(H)=2$ and $H$ has an induced $TK_3$.
    \textbf{(b)} An illustration of the clause gadget, where each vertex of $H$ attached to an enforcer is in the square. The input vertices have been filled in.
    Dashed edges may or may not exist, but their existence is inconsequential to the correctness of our gadget. 
    \textbf{(c)} 
     We show the eight possible combinations of inputs that can be given to the gadget. Observe that a $(P_3, H)$-good coloring is always possible unless the input is three red edges.
    }
    \label{fig:cg-2}
\end{figure}

\begin{theorem}
\label{thm:p3-hard-2}
$(P_3, H)$-Arrowing is coNP-complete when $H$ is a $2$-connected graph on at least four vertices with $\mepl(H) \geq 2$.
\end{theorem}

\begin{proof}
    We follow the same argument as in the proof of Theorem~\ref{thm:p3-hard-1}.
    We first discuss the variable gadget.
    Using Lemma~\ref{lem:mlep2-combine} and the fact that 
    $\mepl(H) \geq 2$, we know that we can construct a graph $F$ with exactly two copies of $H$ that share a single edge, unless $H= J_4$.
    The variable gadget for this exception is shown in Appendix~\ref{app:a}. For now, we assume $H \not= J_4$. 
    Let $A$ and $B$ the copies of $H$ in $F$ and 
    let $(b,c)$ be the edge that was identified. Let $e \not= b$ be a vertex in $B$ adjacent to $c$, and let $(a,d)$ be an edge in $A$ where $a,d \not\in \{b,c\}$.
    Note that if such an $(a,d)$ does not exist in $H$, then $(b,c)$ must be an edge that shares a vertex with every other edge in $H$. However, it is easy to see that in this case two copies of $H$ cannot be identified on $(b,c)$ without forming new copies of $H$. Thus, $(a,d)$ must exist.
    We now append enforcers to each vertex in $A$ and $B$ except $a,b,c,d,$ and $e$.
    Our variable gadget is constructed using two copies of $F$ joined via signal extenders as shown in Figure~\ref{fig:vg-2}.

    We use a clause gadget similar to the one used in Theorem~\ref{thm:p3-hard-1}, but with some modifications. 
    Since $\mepl(H) \geq 2$, we know that $H$ must contain a $C_4$ or a $TK_3$. If $H$ contains a $TK_3$, we can use the gadget shown in Figure~\ref{fig:cg-2}.
    However, if $\mepl(H) = 2$ and we only have induced $C_4$'s, we can use the gadget shown in Figure~\ref{fig:cg-1} when $|V(H)| \geq 5$ using the induced $C_4$ and another vertex $e$. 
    The only case left is $H = C_4$, which is discussed in Appendix~\ref{app:a}. 
\end{proof}

\section{Extending the hardness of \texorpdfstring{$\mathbf{(P_3, H)}$-Arrowing}{(P3, H)-Arrowing}}
\label{sec:extend}

In this section, we discuss how hardness results for $(P_3, H)$-Arrowing can be extended to other $(F,H)$-Arrowing problems. 
We believe this provides an easier method for proving hardness compared to constructing SAT gadgets.
We discuss two methods in which our results can be extended: \textit{(1)} showing that 
$G \ra (P_3, H) \iff G \ra (P_3, H')$ 
for some pairs of $H$ and $H'$ (Section~\ref{sec:tailedkn}), and \textit{(2)} given a graph $G$, showing how to construct a graph $G'$ such that $G \ra (P_3, H) \iff G' \ra (F, H)$ for some $F$
(Section~\ref{sec:stars}). 

\subsection{\texorpdfstring{$\mathbf{P_3}$}{P3} versus tailed complete graphs}
\label{sec:tailedkn}

We first observe that edges not belonging to $H$ can be removed while working on $(P_3, H)$-Arrowing for a graph $G$; we can always color said edge blue without forming a blue $H$.

\begin{observation}
\label{obs:g-minus-e}
    Let $G$ be a graph and $e \in E(G)$. If $e$ does not belong to a copy of $H$ in $G$, then $G \ra (P_3, H)$ if an only if $G - e \ra (P_3, H)$.
\end{observation}

\begin{theorem}
For $n \geq 3$, $G \ra (P_3, TK_n)$ if and only if $G \ra (P_3, K_n)$.
\end{theorem}

\begin{proof}
    Clearly if $G \ra (P_3, TK_n)$ then $G \ra (P_3, K_n)$ since $K_n$ is a subgraph of $TK_n$. For the other direction,
    consider a graph $G$ such that $G \ra (P_3, K_n)$ but is $(P_3, TK_n)$-good. 
    By Observation~\ref{obs:g-minus-e}, we can assume that each edge in $G$ belongs to a $K_n$. We can also assume that $G$ is connected.
    Let $c$ be a $(P_3, TK_n)$-good coloring of $G$. 
    Since $G \ra (P_3, K_n)$ there must exist a blue $K_n$ in $c$. 
    Let $U = \{u_1, u_2, \ldots, u_n\}$ be the vertices of said $K_n$.
    Let $e = (u_i, v)$
    be an edge going from some $u_i \in U$ to a vertex $v \not\in U$. We know that such an edge exists otherwise $G$ is just a $K_n$, and a $K_n$ is $(P_3, K_n)$-good, so this would contradict our assumption. W.l.o.g., let $u_i = u_1$.
    Note that $(u_1, v)$ must be red, otherwise, we have a blue $TK_n$. Since $(u_1, v)$ is part of a $K_n$ (Observation~\ref{obs:g-minus-e}), at least one vertex $w$ must be connected to both $u_1$ and $v$. 
    Note that $(u_1, w)$ and $(v,w)$ must be blue to avoid a red $P_3$. 
    If $w \in U$, then $U$ and $(v,w)$ form a blue $TK_n$. If $w \in V(G) - U$, then $U$ and $(u_1, w)$ form a blue $TK_n$.
\end{proof}

\begin{corollary}
$(P_3, TK_n)$-Arrowing is coNP-complete when $n \geq 4$ and in P when $n=3$.
\end{corollary}

With this result, we have categorized the complexity of all $(P_3, H)$-Arrowing problems for connected $H$ with $|V(H)| \leq 4$; the star and path graphs were shown to be in P~\cite{burr1976graphs,hassan2023}. 

\subsection{Stars versus 
\texorpdfstring{$\mathbf{2}$-connected}{2-connected}
graphs}
\label{sec:stars}

Note that $P_3 = K_{1,2}$. 
Given a graph $G$, suppose we construct a graph $G'$ by taking a copy of $G$ and appending an edge (one for each vertex in $G$) to each vertex, where each appended edge is forced to be red in all colorings.
It is easy to see that if a coloring of $G$ contains a red $K_{1,2}$, then said coloring in $G'$ contains a red $K_{1,3}$, using the appended red edge.
Thus, if we can find a $(K_{1,3}, H)$-good graph $F$ with an edge $(u,v)$ such that,
in all good colorings, $(u,v)$ is red 
and no other edge adjacent to $v$ is red, we could reduce $(K_{1,2}, H)$-Arrowing to $(K_{1,3}, H)$-Arrowing by appending a copy of $F$ (identifying $v$) to each vertex of $G$. 
Recall that we do not have to worry about new copies of $H$ due to its $2$-connectivity.

Generalizing this argument, if we attach $k$ red edges to each vertex of $G$, then a coloring of $G$ with a red $K_{1,\ell}$ corresponds to a coloring of $G'$ with a red $K_{1,k+\ell}$. 
In Appendix~\ref{app:b}, we show that for all $n \geq 3$, there exists some $m$ for which there is a 
$(K_{1,n}, H)$-good graph with a vertex $v$ that is always the center of a $K_{1,m}$ for some $1 \leq m < n$. This allows us to reduce $(K_{1,n-m}, H)$-Arrowing to $(K_{1,n}, H)$-Arrowing. Thus, we can assert the following:

\begin{restatable}[]{lemma}{starhlemma}
\label{lem:k1nh-extend}
Suppose $H$ is a $2$-connected graph.
If $(K_{1,k}, H)$-Arrowing is coNP-hard for all $2 \leq k < n$, then  $(K_{1,n}, H)$-Arrowing is also coNP-hard.
\end{restatable}

Also in Appendix~\ref{app:b}, we reduce $(2,2)$-3SAT to $(K_{1,3}, K_3)$-Nonarrowing and show
how that result can be extended to $(K_{1,n}, K_3)$-Arrowing for $n \geq 4$,
giving us the following result:

\begin{theorem}
    For all $2$-connected $H$ and $n \geq 2$,
    $(K_{1,n}, H)$-Arrowing is coNP-complete with the exception of $(K_{1,2}, K_3)$-Arrowing, which is in P.
\end{theorem}

Finally, recall that $(P_3, H)$-Nonarrowing is equivalent to $H$-free Matching Removal. We can assert a similar equivalence between $H$-free $b$-Matching Removal and $(K_{1,b+1}, H)$-Nonarrowing, giving us the following corollary:

\begin{corollary}
    For all 2-connected $H$, $H$-free $b$-Matching Removal is NP-complete for all $b \geq 1$, except the case where $b=1$ and $H = K_3$, which is in P.
\end{corollary}

\section{Conclusion and future work}
\label{sec:conclude}

This paper provided a complete categorization for the complexity of $(P_3, H)$-Arrowing when $H$ is $2$-connected. We provided a polynomial-time algorithm when $H=K_3$, and coNP-hardness proofs for all other cases. Our gadgets utilized a novel graph invariant, minimum edge pair linkage, to avoid unwanted copies of $H$. We showed that our hardness results can be extended to $(P_3, TK_n)$- and $(K_{1,n}, H)$-Arrowing using easy-to-understand graph transformations. 

Our ultimate goal is to categorize the complexity of all $(F,H)$-Arrowing problems.
Our first objective is to categorize the complexity of $(P_3, H)$-Arrowing for all $H$, and to find more graph transformations to extend hardness proofs between different arrowing problems.

\bibliography{mybib}

\begin{thebibliography}{10}

\bibitem{berman200322sat}
P.~Berman, M.~Karpiński, and A.~Scott.
\newblock {Approximation Hardness of Short Symmetric Instances of MAX-3SAT}.
\newblock {\em ECCC}, 2003.

\bibitem{bikov2018}
A.~Bikov.
\newblock {\em {Computation and Bounding of {F}olkman Numbers}}.
\newblock PhD thesis, Sofia University ``St. Kliment Ohridski'', 06 2018.

\bibitem{Bu3}
S.A. Burr.
\newblock {On the Computational Complexity of Ramsey-Type Problems}.
\newblock {\em Mathematics of Ramsey Theory, Algorithms and Combinatorics}, 5:46--52, 1990.

\bibitem{burr1976graphs}
S.A. Burr, P.~Erd\H{o}s, and L.~Lov{\'a}sz.
\newblock {On graphs of {R}amsey type}.
\newblock {\em Ars Combinatoria}, 1(1):167--190, 1976.

\bibitem{edmonds1965paths}
J.~Edmonds.
\newblock {Paths, Trees, and Flowers}.
\newblock {\em Canadian Journal of Mathematics}, 17:449--467, 1965.

\bibitem{hassan2023}
Z.R. Hassan, E.~Hemaspaandra, and S.~Radziszowski.
\newblock {The Complexity of $(P_k, P_\ell)$-Arrowing}.
\newblock In {\em {FCT} 2023}, volume 14292, pages 248--261, 2023.

\bibitem{lima2017decycling}
C.V.G.C. Lima, D.~Rautenbach, U.S. Souza, and J.L. Szwarcfiter.
\newblock {Decycling with a Matching}.
\newblock {\em Information Processing Letters}, 124:26--29, 2017.

\bibitem{lima2018bipartizing}
C.V.G.C. Lima, D.~Rautenbach, U.S. Souza, and J.L. Szwarcfiter.
\newblock {Bipartizing with a Matching}.
\newblock In {\em COCOA 2018}, pages 198--213, 2018.

\bibitem{DBLP:journals/anor/LimaRSS22}
C.V.G.C. Lima, D.~Rautenbach, U.S. Souza, and J.L. Szwarcfiter.
\newblock {On the Computational Complexity of the Bipartizing Matching Problem}.
\newblock {\em Ann. Oper. Res.}, 316(2):1235--1256, 2022.

\bibitem{ds1}
S.~Radziszowski.
\newblock {Small Ramsey Numbers}.
\newblock {\em Electronic Journal of Combinatorics}, DS1:1--116, January 2021.
\newblock URL: \url{https://www.combinatorics.org/}.

\bibitem{rosta}
V.~Rosta.
\newblock {Ramsey Theory Applications}.
\newblock {\em Electronic Journal of Combinatorics}, DS13:1--43, December 2004.
\newblock URL: \url{https://www.combinatorics.org/}.

\bibitem{rut:c:graph-coloring}
V.~Rutenburg.
\newblock {Complexity of Generalized Graph Coloring}.
\newblock In {\em {MFCS 1986}}, volume 233 of {\em Lecture Notes in Computer Science}, pages 573--581. Springer, 1986.

\bibitem{Scha}
M.~Schaefer.
\newblock {Graph Ramsey Theory and the Polynomial Hierarchy}.
\newblock {\em Journal of Computer and System Sciences}, 62:290--322, 2001.

\end{thebibliography}

\newpage
\begin{appendix}

\section{Proof of Claim~\ref{claim:lem:mlep2-combine} and missing gadgets}
\label{app:a}
\meplclaim*
\begin{claimproof}
\begin{enumerate}
\item This follows from our definition of $Z$.
\item If at most one vertex in $\{u, v\}$ were in $Z$, deleting it would disconnect said copy of $H$, contradicting the fact that $H$ is $2$-connected.
\item Suppose that both $E_{Z_X}$ and $E_{Z_Y}$ are nonempty. Let $f_1 \in E_{Z_X}$ and $f_2 \in E_{Z_Y}$. We have $\epl_{A_{H,e}}(f_1, f_2) = 0$
since, by construction,
$E_{A_{H,e}}(X - Y, Y - X) = \emptyset$.
Since both of these edges belong to $Z$, which is isomorphic to $H$, we also have $\epl_H(f_1, f_2) = 0$, which contradicts our assumption that $\mepl(H) \geq 2$.
\item Note that $k$-connected graphs must have minimum degree at least $k$ as a vertex with fewer neighbors could be disconnected from the rest of the graph with fewer than $k$ vertex deletions. So, we have $d_H(w) \geq 2$ for each $w \in V(H)$.
W.l.o.g., assume that $E_{Z_X} = \emptyset$. Then, each vertex in $w \in Z_X$ can only be connected to $u$ and $v$. Thus, we have $d_H(w) = d_Z(w) \leq 2$, and consequently $d_H(w) = 2$. \claimqedhere
\end{enumerate}
\end{claimproof}

\noindent \textbf{Missing gadgets for \texorpdfstring{$\mathbf{(P_3, H)}$-Nonarrowing}{(P3, H)-Nonarrowing}.}
In Figure~\ref{fig:p3j4} we show the variable gadget for $(P_3, J_4)$-Nonarrowing.
In Figure~\ref{fig:p3c4} we show the clause gadget for $(P_3, C_4)$-Nonarrowing.

\begin{figure}[t]
    \centering
    \includegraphics[width=0.8\textwidth]{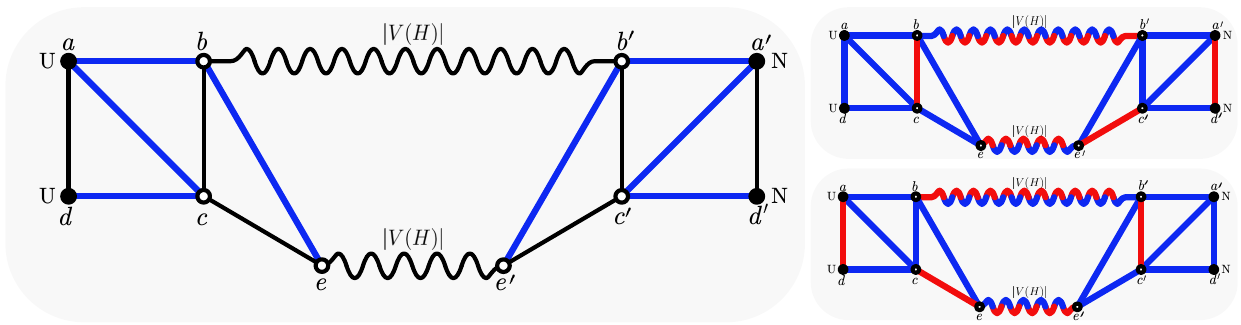}
    \caption{The variable gadget for $(P_3, J_4)$-Nonarrowing is shown on the left. The colorings on the right show that when vertices marked $\mathbf{U}$ (resp., $\mathbf{N}$) are free, those marked $\mathbf{N}$ (resp., $\mathbf{U}$) are nonfree.
    }
    \label{fig:p3j4}
\end{figure}

\begin{figure}[t]
    \centering
    \includegraphics[width=0.65\textwidth]{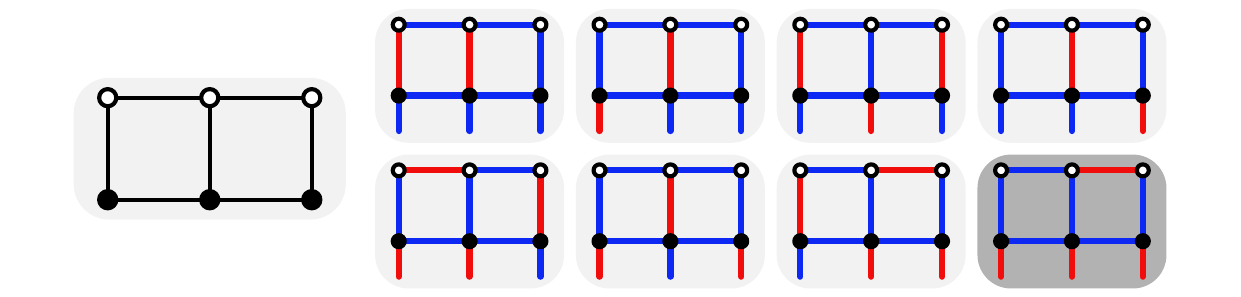}
    \caption{The clause gadget for $(P_3, C_4)$-Nonarrowing is shown on the left. Possible inputs for the gadget are shown on the right. A good coloring is possible unless the input is three red edges.
    }
    \label{fig:p3c4}
\end{figure}

\section{Extending to stars}
\label{app:b}

We first define the following special graph which will be used in our extension proof.

\begin{definition}
A graph $G$ is called an
    $(F, H)$-leaf sender
    with leaf-signal edge $(u,v)$ 
    if it is $(F,H)$-good, 
    $(u,v)$ is red in all good colorings, 
    and there exists a good coloring where 
    $(u,v)$ is not adjacent to any other red edge.
\end{definition}

When the context is clear,
we will use the shorthand \textbf{append a leaf sender to $u \in V(G)$} to mean we will 
add an $(F, H)$-leaf sender to $G$ and identify a vertex of its leaf-signal edge with $u$. 
This graph essentially simulates ``appending a red edge'' as described in Section~\ref{sec:stars}.

\subsection{Proof of Lemma~\ref{lem:k1nh-extend}}
\starhlemma*
\begin{proof}

Suppose we are trying to prove the hardness of $(K_{1,n},H)$-Arrowing for $2$-connected $H$ and $n \geq 3$.
Let $A$ be a ``minimally bad'' graph such that $A \ra (K_{1,n}, H)$, but removing any edge $e$ from $A$ gives a $(K_{1,n},H)$-good graph. 
Let $e = (u, w)$ and $A' = A - e$. 
Let $\mathcal{C}$ be the set of all good colorings of $A'$.
For a coloring $c \in \mathcal{C}$ and vertex $v$, let $r_c(v)$ be the number of red edges that $v$ is adjacent to in $c$. 
Let $r_{\mathcal{C}}(v) = \min_{c \in \mathcal{C}} r_c(v)$.
We consider different cases for $r_{\mathcal{C}}(u)$. Since $c$ is a good coloring, we know that 
$r_{\mathcal{C}}(u) \leq n-1$.
\begin{enumerate}
    \item $r_{\mathcal{C}}(u) = n-1$. 
    In this case, it is easy to see that $A'$ is a $(K_{1,n}, H)$-enforcer with signal vertex $u$.
    Let $(p,q) \in V(H)$. 
    Construct a graph $B$ like so.
    Take a copy of $H$ and append an enforcer to each vertex of $H$ except $p$ and $q$. 
    It is easy to see that $(p,q)$ must be red in all good colorings, i.e., $B$ is a $(K_{1,n}, H)$-leaf sender with leaf-signal edge $(p,q)$. 
    For any graph $G$, we 
    append a leaf-sender to each of $G$'s vertices to obtain a graph $G'$ such that $G' \ra (K_{1,n}, H) \iff G \ra (K_{1,{n-1}}, H)$, as discussed in Section~\ref{sec:stars}.
    
    \item $r_{\mathcal{C}}(u) = 0$. 
    Let $m = |V(H)|$.
    In this case, we can construct a $(K_{1,n}, H)$-enforcer, $B$, by combining $2m$ copies of $A'$ as we did in Lemma~\ref{lem:p3enforce}. 
    Note that in any good coloring $c$ of $A'$, we have that $r_c(u) \geq n-1$ or 
    $r_c(w) \geq n-1$; if not, such a coloring and a red $(u,w)$ gives a good coloring for $A$, contradicting the fact that $A \ra(K_{1,n}, H)$.
Make $2m$ copies of $A'$, where $u_i$ (resp., $w_i)$ refers to the vertex $u$ (resp., $w$) in the  $i^{\text{th}}$ copy of $A'$, referred to as $A'_i$.
Now, identify each $w_i$ with $u_{i+1}$ for $i \in \{1,2,\ldots,2m-1\}$, and identify $w_{2m}$ with $u_1$. 
Observe that when $w_1$ is adjacent to $n-1$ red edges in $A'_1$, then $u_2$ cannot be adjacent to any red edge in $A'_2$, causing $w_2$ to be adjacent to $n-1$ red edges in $A'_2$, and so on. 
Since every $u_i$ and $w_i$ is adjacent to $n-1$ red edges, any of them can be our signal vertex $v$.
We can now proceed as we did in the previous case to reduce from $(K_{1,n-1}, H)$-Arrowing. 
    
    \item $r_{\mathcal{C}}(u) \in \{1,2,\ldots, n-2\}$. 
    For any graph $G$, we can attach a copy of $A'$ to each of $G$'s 
    vertices---identifying $u \in V(A')$ with each vertex---to obtain a graph $G'$ such that $G' \ra (K_{1,n}, H) \iff G \ra (K_{1,{n-r_{\mathcal{C}}(u)}}, H)$, as discussed in Section~\ref{sec:stars}, thereby providing a reduction from $(K_{1,{n-r_{\mathcal{C}}(u)}}, H)$-Arrowing. \qedhere
\end{enumerate}
\end{proof}

\begin{figure}[t]
    \centering
    \includegraphics[width=0.7\textwidth]{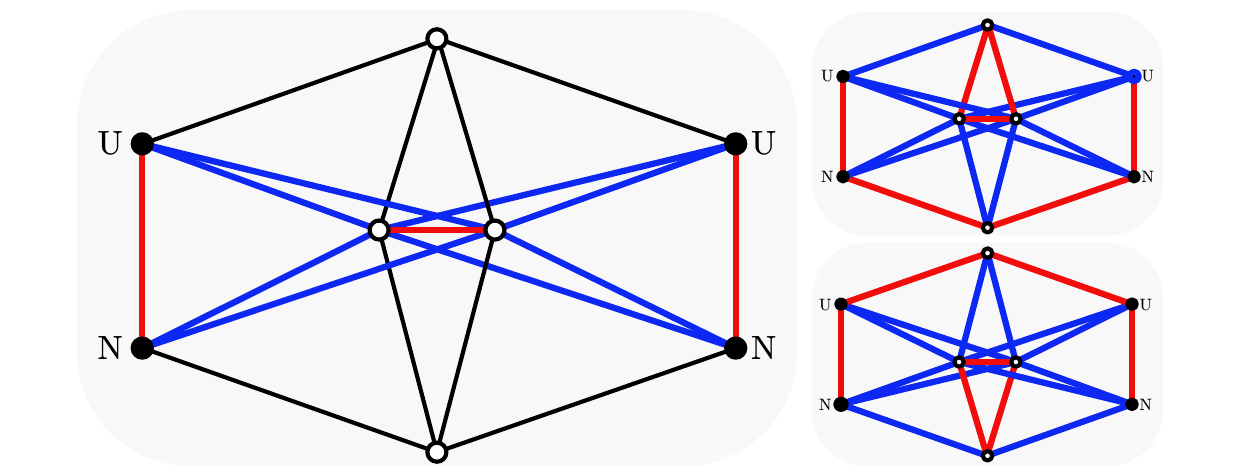}
    \caption{The variable gadget for $(K_{1,3}, K_3)$-Nonarrowing is shown on the left. Edges with the same color in all good colorings have been pre-colored. 
    Both good colorings are shown on the right.
    }
    \label{fig:k13k3-var}
\end{figure}

\subsection{Hardness of \texorpdfstring{$\mathbf{(K_{1,n}, K_3)}$-Nonarrowing}{(K-1-n, K3)-Nonarrowing} for \texorpdfstring{$n \geq 3$}{n >= 3}
}

To show that $(K_{1,3}, K_3)$-Arrowing is coNP-complete, we provide gadgets as we did for $(P_3,H)$-Arrowing.
We provide gadgets in Figures~\ref{fig:k13k3-var} and~\ref{fig:k13k3-clause} to show that 
$(2,2)$-3SAT can be reduced to 
$(K_{1,3}, K_3)$-Nonarowing. 
Note that the output vertices are either attached to a single red edge or two red edges. When they are attached to a single red edge, they behave like true output signals. When adjacent to two red edges, they behave like a false output signal. 
The clause gadget behaves like an OR gate, in that it has no good coloring when the three input vertices all have false inputs.

\begin{figure}[t]
    \centering
    \includegraphics[width=0.7\textwidth]{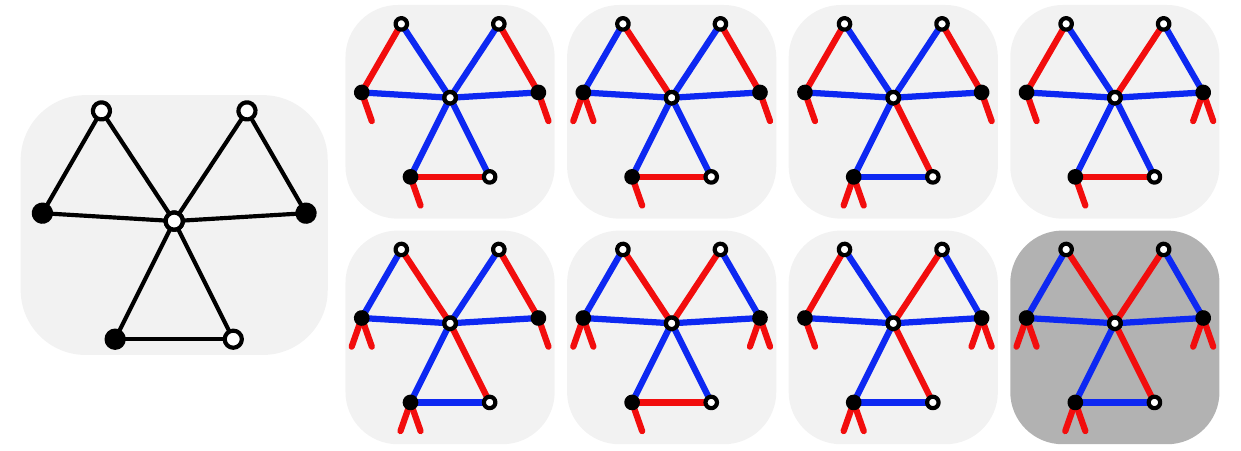}
    \caption{The clause gadget for $(K_{1,3}, K_3)$-Nonarrowing is shown on top. The eight combinations of inputs that can be given to the gadget are shown on the bottom. Observe that a $(K_{1,3}, K_3)$-good coloring is always possible unless the input is three red $K_{1,2}$'s.
    }
    \label{fig:k13k3-clause}
\end{figure}
\begin{figure}[t]
    \centering
    \includegraphics[width=0.7\textwidth]{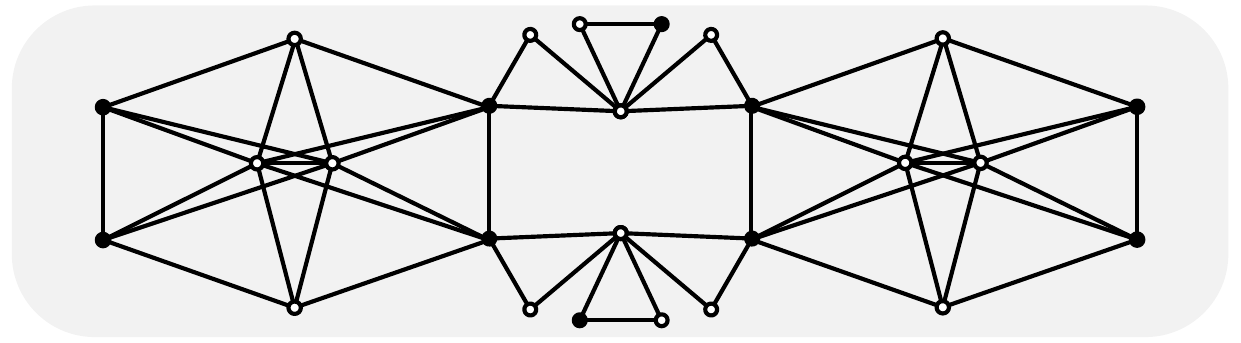}
    \caption{The smallest cycle made when joining $(K_{1,3},K_3)$ variable and clause gadgets is a $C_6$.
    }
    \label{fig:k1k3-join}
\end{figure}
\begin{figure}
    \centering
    \includegraphics[width=0.8\textwidth]{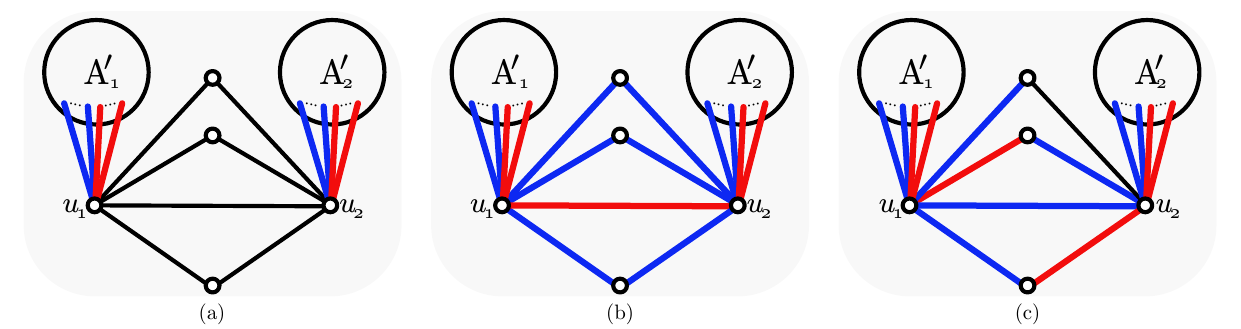}
    \caption{$A'_1$ and $A'_2$ are graphs with $r_{\mathcal{C}}(u_i) = n-2$. We show how to construct a new graph using these in \textbf{(a)}. In \textbf{(b)}, we show a good coloring where $u_i$'s are now adjacent to $n-1$ red edges. Finally, in \textbf{(c)}, we observe that the coloring in \textbf{(b)} is the only good coloring since at most one edge from outside $A'_i$ that is adjacent to $u_i$ can be red. 
    }
    \label{fig:k13k3-spesh}
\end{figure}

Recall that in the hardness proofs for $(P_3, H)$-Arrowing, we also had to show that no new $H$ is constructed while combining gadget graphs to construct $G_{\phi}$. It is easy to see that no new $K_3$ is constructed when our gadgets are combined; since each clause has unique literals, a cycle formed while constructing $G_\phi$
would have to go through at least two clause gadgets and at least two variable gadgets, but this cycle has more than three vertices (see Figure~\ref{fig:k1k3-join}).

To show the hardness of $(K_{1,n}, K_3)$-Arrowing for $n \geq 4$, we can proceed exactly as we did in Lemma~\ref{lem:k1nh-extend}. 
The only case where the proof fails is when 
$r_{\mathcal{C}}(u) = n-2$, because now the proof says we have to reduce $(K_{1,2}, K_3)$-Arrowing to $(K_{1,n},K_3)$-Arrowing, which is unhelpful since $(K_{1,2},K_3)$-Arrowing is in P. In Figure~\ref{fig:k13k3-spesh} we show how vertices attached to $n-2$ red edges can be combined to make a $(K_{1,n}, K_3)$-enforcer. Using the enforcer, we can create a $(K_{1,n}, K_3)$-leaf sender and reduce from $(K_{1,n-1}, K_3)$-Arrowing.

\end{appendix}

\end{document}